\documentclass[conference, 10pt]{IEEEtran}

\newcommand{\isrevision}{1}
\usepackage{graphicx}
\usepackage{caption}
\usepackage{subcaption}
\usepackage{url}
\graphicspath{ {./images/} }
\usepackage{amsthm}
\usepackage{amssymb} % for \triangleq
\usepackage{amsmath} % for \text
\newcommand{\defeq}[0]{\triangleq}
\usepackage{xspace}
\usepackage{multicol}
\usepackage[T1]{fontenc}
\usepackage[utf8]{inputenc}
\usepackage{authblk}

\usepackage{lipsum}
\usepackage{mathtools}
\usepackage{cuted}
\usepackage{authblk}

\usepackage{paralist}
\setdefaultleftmargin{0em}{0em}{0em}{0em}{0em}{0em}
\usepackage{enumitem}

\makeatletter
\renewcommand\AB@affilsepx{, \protect\Affilfont}
\makeatother

\usepackage{tabularx}

\usepackage[normalem]{ulem} % For strikeout

\makeatletter
\def\url@allbreakstyle{%
  \def\UrlBreaks{\do\.\do\@\do\\\do\/\do\!\do\_\do\|\do\;\do\>\do\]%
    \do\)\do\,\do\?\do\'\do+\do\=\do\#%
    \do A\do B\do C\do D\do E\do F\do G\do H\do I\do J\do K\do L\do M%
    \do N\do O\do P\do Q\do R\do S\do T\do U\do V\do W\do X\do Y\do Z%
    \do a\do b\do c\do d\do e\do f\do g\do h\do i\do j\do k\do l\do m%
    \do n\do o\do p\do q\do r\do s\do t\do u\do v\do w\do x\do y\do z%
    \do 0\do 1\do 2\do 3\do 4\do 5\do 6\do 7\do 8\do 9%
  }%
}
\def\url@restrictedbreakstyle{%
  \def\UrlBreaks{\do\.\do\@\do\\\do\/\do\!\do\_\do\|\do\;\do\>\do\]%
    \do\)\do\,\do\?\do\'\do+\do\=\do\#}%
}
\makeatother

\ifnum\isrevision=1
    \usepackage[colorinlistoftodos,textsize=scriptsize]{todonotes} % todo
    
    \newcommand{\stkout}[1]{\ifmmode\text{\sout{\ensuremath{#1}}}\else\sout{#1}\fi}
    
\else
    \usepackage[colorinlistoftodos,textsize=scriptsize, disable]{todonotes}
    
\fi

\newtheorem{theo}{Theorem}

\usepackage{hyperref}
\DeclareMathOperator*{\argmax}{arg\,max}

\begin{document}
\title{Coalitional Game-Theoretical Approach to Coinvestment with Application to Edge Computing}

\author{
Rosario Patanè \\ Université Paris-Saclay \\ \emph{rosario.patane@upsaclay.fr}
\\ \\
Diego Kiedanski \\ Yale University \\ \emph{diego.kiedanski@yale.edu} 
\and Andrea Araldo; Tijani Chahed \\ Télécom SudParis - Institut Polytechnique de Paris \\ \emph{<firstname>.<lastname>@telecom-sudparis.eu}
\\ \\
Daniel Kofman \\ Télécom Paris - Institut Polytechnique de Paris \\ \emph{daniel.kofman@telecom-paristech.fr}
}

\maketitle
\begin{abstract}
We propose in this paper a coinvestment plan between several stakeholders of different types, namely a physical network owner, operating network nodes, e.g. a network operator or a tower company, and a set of service providers willing to use these resources to provide services as video streaming, augmented reality, autonomous driving assistance, etc. One such scenario is that of deployment of Edge Computing resources. %Low latency applications like augmented reality or autonomous driving cannot tolerate the delays imposed by the cloud computing paradigm. To enable them, the Edge Computing paradigm consists in deploying computational resources close to the final users. 
Indeed, although the latter technology is ready, the high Capital Expenditure (CAPEX) cost of such resources is the barrier to its deployment. 
For this reason, a solid economical framework to guide the investment and the returns of the stakeholders is key to solve this issue.
We formalize the coinvestment framework using coalitional game theory. We provide a solution to calculate how to divide the profits and costs among the stakeholders, taking into account their characteristics: traffic load, revenues, utility function. We prove that it is always possible to form the grand coalition composed of all the stakeholders, by showing that our game is convex. We derive the payoff of the stakeholders using the Shapley value concept, and elaborate on some properties of our game. We show our solution in simulation.
\end{abstract}


\IEEEpeerreviewmaketitle

\section{Introduction}
\label{sec:introduction}

Coinvestment allows several stakeholders to share expenses and revenues when deploying of certain projects, otherwise non-beneficial for them.
We model coinvestment via coalitional game theory. Our model can be applied in scenarios when (i)~costly resources must be deployed in some nodes, (ii)~only one entity, hereafter named Network Owner (NO), has access to these nodes and (iii)~such resources are beneficial to third party Service Providers (SPs).
One such scenario is Edge Computing (EC), which is the main application of this paper.
In EC, physical nodes at the edge are possessed by an NO, which can be a network operator, like AT\&T, or a tower company~\cite{Guillemin2021}; in such nodes, the NO deploys computational resources,\footnote{
Cloud providers, e.g. Amazon, distribute ``edge'' resources for rent, but their edge  locations go as far as ``data-centers at the edge of the 5G network''(\url{aws.amazon.com/fr/edge}). We instead consider edge nodes much closer to users, e.g. base stations or road side units, which are owned by the NO. 
%The NO is thus the one that can deploy computational resources that far in the edge.
}
 which are used by third party SPs, e.g. video streaming services, as YouTube, Netflix, or car manufacturers offering in the future automated driving services, as Tesla or Renault~\cite{AraldoSAC2020,Araldo2016}. The SPs use the resources at the Edge, closer to their end-users, to distribute the load or to satisfy their low latency applications. 

The deployment costs however can be very high and cannot be supported solely by the NO. This explains today's impediment for the widespread of EC.
The deployment costs however can be very high and cannot be supported solely by the NO.
 This explains today's impediment for the widespread of EC. On the other hand, SPs offering new services, can make large benefits by offering good quality services to their end-users. Hence, it is reasonable to assume that SPs may be willing to contribute to the cost of the deployment of EC resources together with the NO~\cite{Sabella2019}. 
 On the other hand, SPs offering new services, can make large benefits by offering good quality services to their end-users. Hence, it is reasonable to assume that SPs may be willing to contribute to the cost of the deployment of EC resources together with the NO~\cite{Sabella2019}. 

The goal of this paper is to understand how coinvestment can occur, i.e. how cost of deployment and benefits should be shared among the NO and SPs. 
Our contributions are:

\begin{compactitem}
    \item We propose a model based on coalitional game theory, and describe the discrepancy between the two categories of stakeholders in terms of revenues (brought by SPs) and infrastructure deployment (by the NO).
    \item We assess the existence of the core, i.e. the fact that the grand coalition composed of all the stakeholders exists and is stable, by showing the convexity of our game.
    \item We propose sharing of the payoff between the stakeholders based on the Shapley value solution concept, which lies in the core in this case.
    \item We evaluate the performance of our proposal, through numerical examples for different scenarios and configurations, considering the case of EC.
\end{compactitem}

The remainder of the paper is as follows. \S\ref{sec:relatedwork} introduces the related work. In \S\ref{sec:model}, we formulate our model. \S\ref{sec: analysis} provides an analysis for the coinvestment plan between the stakeholders. 
In \S\ref{sec: numerical results}, we show numerical results for several scenarios and configurations. \S\ref{sec: conclusion}  concludes the paper.

\section{Related Work}
\label{sec:relatedwork}
The advantage of sharing resources between several tenants has been shown in~\cite{Afraz2018}, for the case of passive optical networks. However, they assume resources are already deployed and allocation to tenants is determined via auctions. We instead consider the case where no resources are deployed yet (as for EC) and thus devise a mechanism for all the players to coinvest to buy and deploy resources, jointly taking into account the resource allocation among players.
An example of coinvestment is given from~\cite{Kiedanski2020}. Investors are producers/consumers of energy, and a battery exists that must be sized to meet everyone's needs. Each of them can decide whether to charge or discharge the battery to use stored energy later, when energy is more expensive (daily hours). The result of coinvestment is the price that each one has to pay to buy the battery. The problem is solved using coalitional game theory~\cite{Tamer2009}, as we shall do in this work. However, we cannot directly use these results, because~\cite{Kiedanski2020} is a linear production problem, valid to model the cost of electricity, while in our case the benefit of using resources at the Edge can be non-linear, due to the diminishing return in resource deployment at the Edge~\cite{Salchow2008}. Another peculiarity of our study is that we have a veto player, which corresponds to the NO. In~\cite{Sereno2012} a coinvestment between Network Providers (NPs) is realized to improve energy efficiency in cellular access networks, i.e. in some parts of the day, except the peak hours, the resources of a single NP can be oversized, so, just a subset of NPs can serve all the load allowing the others to save energy. The resultant benefit of this sharing is later divided among the players. This paper is very close to the problem in this work, but does not match exactly because: (i) the players are at the same level (all NPs), while in our case players have different roles (NO and SPs), (ii) in our problem we do not care about exchanging load between players to save energy. Coalitional Game Theory has been applied to EC~\cite{Moura2019}, e.g. for pricing, spectrum or content sharing in D2D communication, task offloading~\cite{Pham2019}. To the best of our knowledge we are the first to apply it for co-investment in EC, and thus there are no other works with which we can compare our proposal.

\section{Coinvestment model}
\label{sec:model}
Our coinvestment problem is modeled as a coalitional game with transferable payoff, i.e. players can share a common amount of utility/cost~\cite{Osborne1994}. The game is defined by the tuple $(\mathcal{N}, v)$, where $\mathcal{N}$ is the set of players and the~\emph{coalitional value} $v(\mathcal{S})$ is a function that associates a value to any subset $\mathcal{S}\subseteq\mathcal{N}$, called \emph{coalition}. 
The players are one physical Network Owner (NO) and $N$ Service providers (SPs). The NO is the entity that owns the network nodes. It could be a network operator, who owns the location of an antenna or a central office. It could be a separated tower company~\cite{Guillemin2021}. The set of players is $\mathcal{N}=\{1,\dots,N, NO\}$.
The first question to be asked in coalitional games is whether the grand coalition $\mathcal{N}$, formed by all players, is stable, i.e. all the players have an incentive to be part of it, in terms of individual payoff. 
The payoff of any player $i\in\mathcal{N}$ is
$x^i=r^i-p^i$, where $r^i$ and $p^i$ are the revenues and the payment, i.e. the capital cost, respectively.

We assume that the NO is willing to host physical resources on its nodes. To fix ideas, such nodes are edge nodes and resources are CPU in our case. They might also refer to  GPU, etc. 
After the deployment, the NO virtualizes and allocates the resources to SPs. SPs do not own physical resources, only the NO does, in base stations for instance. Whenever players form a coalition, they invest together in deploying an amount of computational capacity $C$, measured in \emph{millicores}. 
Denoting by $h^i$ the resources allocated to player $i$, $\sum_{i \in S}h^i = C$. For $d$ denoting the price expressed in \emph{dollars} per resource unit, the sum of all players' payments should be such that $\sum_{i \in S}p^i=d\cdot C$. As the NO is only player that hosts the capacity, if it does not join the coalition, the coinvestment cannot take place, and so, no computational capacity can be deployed and no EC can be realized. So, for any coalition $\mathcal{S}$, capacity is subject to:

\begin{equation}
\label{eq:C}
C=\left\{\begin{aligned}
\frac{1}{d } \cdot \sum_{i \in S}p^i\ \text{if}\ \text{NO} \in \mathcal{S}\\
0 \ \text{otherwise}
\end{aligned}\right . 
\end{equation}

At every timeslot $t$, each player $i$ has an expected load, $l_t^i$, i.e. the average number of requests coming from users of SP $i$ at time $t$, and which is exogenous to the problem. We assume that the average load profile is the same every day. Each player $i$ has an instantaneous \textbf{utility} $u^i_t$, in monetary units, e.g. \emph{dollars}, which represents the revenues coming from the end users of the services. In other words, the utility is what users pay to a SP to consume its services.
%\footnote{In this problem, the load and the shape of the utility function are inputs and assumed to be known and truthful.}
We assume the load and the shape of the utility function are known and truthful inputs.
Similarly to~\cite{Misra2015}, the utility is function of the expected load $l^i_t$ (the more users consume the service, the more they pay) and of the allocated resources $h^i$ (the more resources, the better the service, the more users are willing to pay):
    \begin{align} 
    \label{eq:null_resources_null_utility}
        u^i_t=u^i(l^i_t,h^i);\ \ \ \ \ u^i(l_t^i,0)=0
    \end{align}
We thus assume (Eqn.\ref{eq:null_resources_null_utility} on the right) that the utility is null if no resources are allocated.

Observe that the case where NO uses some Edge resources for itself can be easily modeled in our framework by introducing a fictitious SP representing the NO using resources of EC.
The NO does not offer any Edge service to the end users, so, its load is null $(l_t^\text{NO}=0)$ and thus it does not need Edge resources for itself ($h^\text{NO}=0$), which implies that $u_t^\text{NO}=0$ and $ \sum_{i \in \mathcal{S}\setminus\{\text{NO}\}}h^i = C$. However, the NO gets a fraction of the value of the grand coalition, if it exists.

The revenues of a coalition is the sum of the utilities of the SPs over the investment period: $\sum_{i \in \mathcal{S}}r^i= D \cdot \sum_{i \in \mathcal{S}}\sum_{t \in [T]} u^i_t$, where $D = 365 \cdot Y$, $Y$ is the duration in years of the investment, $T=96$ is the number of timeslots in one day, considering each timeslot has duration $15 \ minutes$. 

We now define the value $v(\mathcal{S})$ of any coalition $\mathcal{S}\subseteq\mathcal{N}$. When forming coalition $\mathcal{S}$, the players involved choose allocation vector $\vec h$ and capacity $C$ so as to maximize the coalition value $v(\mathcal{S})$:
\begin{align}
\label{eq:max1}
    %v^{\Vec{h}, C, \mathcal{S}} \defeq 
     v(\mathcal{S}) 
     &= \max_{\vec h, C} v^{\Vec{h}, C,\mathcal{S}} 
     \defeq 
     \max_{\vec h, C} D \sum_{i \in \mathcal{S}}\sum_{t=1}^T u^i(l_t^i, h^i) 
     - d \cdot C
\\
\text{s.t.} &
    \label{eq: capacity_constr}
    \sum_{i \in \mathcal{S}\setminus\{\text{NO}\}}h^i=C; \,\,\,\, h^\text{NO}=0.
\\
\label{eq:variables_constr}
        & C,  h^i\geq 0, \ \forall t \in [T], \forall i \in \mathcal{S}.
\end{align}
Observe that NO is a \textbf{veto player}. Indeed, withouth NO, the investment does not take place~\cite[\S13.2]{Osborne1994} (see~\eqref{eq:C}). Therefore, applying~\eqref{eq: capacity_constr}, \eqref{eq:null_resources_null_utility} and~\eqref{eq:max1}, if $\text{NO}\notin\mathcal{S}$, we get $v(\mathcal{S})=0$. We will see also that the SPs form altogether a veto player too.

\section{Analysis}
\label{sec: analysis}
We now assess the cooperative structure of our game and its stability. We show the existence of the core, and hence the formation of the grand coalition.

\subsection{Core and convexity of the game}
\label{sec: Core}
Let us define a payoff vector $(x^i)_{i \in \mathcal{N}}$.
The core is a set of payoff vectors, such that a payoff vector is in the core if the payoffs of each player are such that no subgroup can gain by quitting the grand coalition and forming a different coalition~\cite{Osborne1994}. A well known result in coalitional game theory affirms that the core is non empty if the game is convex~\cite{Shapley1971} and that a particular payoff vector having some ``fairness'' properties, i.e. the Shapley value~\cite{Tamer2009,Shapley1971} lies in the core. 

\begin{theo} 
\label{theo: the game is convex}
(The game is convex). 
Our game $(\mathcal{N}, v)$, whose value function $v$ is described by the optimization problem~\eqref{eq:max1}-\eqref{eq:variables_constr} is convex.
\end{theo}
\begin{proof} 
We can rewrite $v(\mathcal{S})$, via~\eqref{eq:max1},\eqref{eq: capacity_constr}, as 
\begin{align}
    \label{eq: separated_value_function}
    v(\mathcal{S}) & =\sum_{i \in \mathcal{S}}\max_{h^i}v^{h^i}
    &
   \text{ where }
   &
   v^{h^i} & \triangleq D \cdot \sum_{t=1}^T u^i(l_t^i, h^i)- d \cdot h^i.
\end{align}

The maximum over $h^i$ of Eqn.~\eqref{eq: separated_value_function} gives the contribution of a single player $i$ to the coalitional value, considering its part of the revenues due to its utility function, and the cost of the resources $h^i$ it uses to produce this utility. 
Observe that such a contribution is independent from the coalition $\mathcal{S}$ in which $i$ participate.
Eqns.~\eqref{eq: separated_value_function} show that the contribution is separable, i.e. it is the summation of the value functions of the individual players. 

We now prove that the game is supermodular, which implies its convexity, thanks to~\cite{Driessen1988}. A game is supermodular if
\begin{align}
\label{eq: check_supermodularity}
    \Delta_i(\mathcal{T}) 
    &\leq \Delta_i(\mathcal{S}), 
    \forall \mathcal{T}\subseteq \mathcal{S}\subseteq \mathcal{N}\setminus\{i\}, \ \forall i \in \mathcal{N}
    \\
\text{where }
    \Delta_i(\mathcal{S}) 
    &= 
    v(\mathcal{S} \cup  \{ i \}) - v(\mathcal{S})
    \underset{\eqref{eq: separated_value_function}}{=}
    \max_{h^i} v^{h^i}
\end{align}
is the marginal contribution of player $i$ to coalition $\mathcal{S}$.

Let us fix any $i\in\mathcal{N}$. Given two coalitions, $\mathcal{S}$ and $\mathcal{T}$, such that $\mathcal{T}\subseteq\mathcal{S}\subseteq\mathcal{N}\setminus\{i\}$, we calculate the marginal contribution of  player $i$ to both coalitions. Consider the case in which $\text{NO} \in \mathcal{T}$ and $i=\text{SP}^i$, for coalition $\mathcal{T}$ we have (see~\eqref{eq: separated_value_function}): $    \Delta_i(\mathcal{T})=
    \max_{h^i} v^{h^i}$.
For coalition $\mathcal{S}$ we have: $\Delta_i(\mathcal{S})=\max_{h^i}v^{h^i}$
so, the marginal contributions are  $\Delta_i(\mathcal{S})=\Delta_i(\mathcal{T})$. 

Now, we consider the case $i=\text{NO}$. In this case the proof is trivial, in fact $\text{NO} \notin \mathcal{T}\cup \mathcal{S}$, so, for the fact that NO is a veto player
\begin{align}
\label{eq: step_00}
            v(\mathcal{T} \cup \{i\}) 
            &\geq 0,  v(\mathcal{T})=0, \ \forall \mathcal{T}\setminus\{\text{NO}\}
\\
        \label{eq: step_01}
         \text{and }  v(\mathcal{S} \cup \{i\}) 
         & \geq 0,  v(\mathcal{S})=0, \ \forall \mathcal{S}\setminus\{\text{NO}\}.
\end{align}
Therefore, Eqn.~\eqref{eq: check_supermodularity} is verified if and only if
\begin{equation}
           v(\mathcal{S} \cup \{i\}) - v(\mathcal{T} \cup \{i\})  \geq 0
\end{equation}
which is equivalent to $\sum_{j \in \mathcal{S}\setminus \mathcal{T}}\max_{h^j}v^{h^j} \geq 0$
where the last inequality is obviously true, since we have a sum of non-negative terms. Another case is the following: $\text{NO} \notin \mathcal{T}\cup \mathcal{S}$ and $i=\text{SP}^i$. This case is easy to prove because we get
$
    v(\mathcal{T} \cup \{i\}) - v(\mathcal{T}) = v(\mathcal{S} \cup \{i\}) - v(\mathcal{S})=0 
$,
which satisfies the definition of supermodularity. The last case to prove is: $i=\text{SP}^i$, $\text{NO} \notin \mathcal{T}$. In this case the marginal contribution of $i$ to coalition $\mathcal{T}$ is null and Eqn.~\eqref{eq: check_supermodularity} becomes
$
           \Delta_i(\mathcal{S}) \geq 0
$. To verify this, we observe that, thanks to~\eqref{eq: separated_value_function}: 
\[
    \Delta_i(\mathcal{S})=v(\mathcal{S}\cup\{i\})-v(\mathcal{S})=
    \begin{cases}\max_{h^i} v^{h_i} & \text{if NO}\in\mathcal{S} \\ 0 & \text{otherwise}\end{cases}\ge 0.
\]
This completes the supermodularity proof and thus convexity. So, the grand coalition can be always formed.
\end{proof}

\subsection{Shapley value}
Finding a mechanism to share the payoff among players is not trivial. One idea would be to divide the payoff equally among players. However, this would not be accepted, since some players contribute to the coalition more than others. First, the NO is a veto player, and its contribution is of primary importance. Second, the SPs do not contribute equally to the coalition: some SPs have more users than others. A second idea would be to share payoffs proportionally to the request load of each SP. However, this would be still unfair, as the benefits collected by SPs do not only depend on the quantity of requests, but also on their type (see \S\ref{sec:utility-function-and-price}).

Fortunately, the proof of convexity in the Th.~\ref{theo: the game is convex} gives us the certainty that
there is a somehow fair way to share the payoff: the Shapley value, which considers the marginal contribution of each player to all the possible coalitions and is computed as~\cite{Osborne1994}: $x^i=\phi^i=\frac{1}{|\mathcal{N}|!}\sum_{\mathcal{S} \subseteq \mathcal{N}\setminus\{i\}}|\mathcal{S}|! \cdot (|\mathcal{N}|-|\mathcal{S}|-1)!\cdot\Delta_i(\mathcal{S})$.

\subsection{Initial investment of players}
\label{sec:initialinvest}
Now that we derived the payoff $x^i$ for each player, we need to calculate how much each player must pay at the beginning of the investment, i.e. $p^i$. This is obtained by solving the following equations:

\begin{flalign}
\label{eq:payments1}
    r^i-p^i &=x^i, \forall i  \in  \mathcal{N}
\\
\text{s.t.} \ \   \sum_{i \in \mathcal{N}} r^i &= D \cdot \sum_{i \in \mathcal{N}}\sum_{t=1}^T u^i(l_t^i, h^{*i})
\\
\label{eq:payments3}
\text{where:} \ \   \vec h^*,C^* &= \argmax_{\vec h, C} v^{\vec h, C,\mathcal{N}} \ \text{s.t.\eqref{eq: capacity_constr}-\eqref{eq:variables_constr}}.
\end{flalign}

\subsection{Relevant properties of our game}
\label{sec:properties of the model}
If player $i$ does not produce revenues and makes not payments, then it is a null player, i.e. $v(\mathcal{S} \cup \{i\})=v(\mathcal{S})$~\cite{Brink2007}. Note that there can be players that do not pay or are even paid ($p^i\le 0$), which still positively contribute to the coalition. For instance, any SP $i$ can positively contribute to the coalition collecting large revenues $r^i$. The NO is never null player, because it is veto player and contributes always to any coalition.

\begin{theo}(Payoff sharing)
\label{theo: payoff_proportions}
The Shapley outcome of the game $(\mathcal{N}, v)$, where $v$ is described by the problem~\eqref{eq:max1}-\eqref{eq:variables_constr}, is divided equally between the NO and the set of all SPs.
\end{theo}
\begin{proof}
Consider the Shapley value of the game $(\mathcal{N}, v)$, i.e. the payoff vector $(\phi_i)_{ i \in \mathcal{N}}$. 
We want to prove that the Shapley value of the NO is equal to the sum of the Shapley values of all SPs. To calculate the Shapley value, we need the value of the marginal contribution of any player $i$ to the coalition, i.e. $\Delta_i(\mathcal{S})= \max_{h^i} v^{h^i}$. The NO is a veto player, and so, the $v$ function is null for coalitions without it. 
\begin{multline}
\label{eq:deltaNO}
    \Delta_\text{NO}(\mathcal{S})= v(\mathcal{S} \cup \{\text{NO}\}) =
    \\\sum_{i \in \mathcal{S}} \max_{h^i} v^{h^i}=\sum_{i \in \mathcal{S}}\Delta_j(\mathcal{S}),
    \forall \mathcal{S}\subseteq\mathcal{N}\setminus\{\text{NO}\}.
\end{multline}

Now, we can show that the Shapley value of the NO is equal to the sum of the SPs Shapley values. We know that the Shapley value is in the core, which is subject to the efficiency property, $\sum_{i \in \mathcal{N}}\phi_i = v(\mathcal{N})$. Hence, $v(\mathcal{N})= \phi_\text{NO} + \sum_{ j \in \mathcal{N}\setminus\{\text{NO}\}}\phi_j$. The Shapley value of the NO is
\begin{multline}
 \phi_\text{NO}=\frac{1}{|\mathcal{N}|!}\sum_{\mathcal{S} \subseteq \mathcal{N}\setminus\{NO\}}|\mathcal{S}|!\cdot (|\mathcal{N}|-|\mathcal{S}|-1)!\cdot\Delta_{\text{NO}}(\mathcal{S}) \underset{\eqref{eq:deltaNO}}{=} \\ \frac{1}{|\mathcal{N}|!}\sum_{\mathcal{S} \subseteq \mathcal{N}\setminus\{NO\}}|\mathcal{S}|!\cdot (|\mathcal{N}|-|\mathcal{S}|-1)!\cdot\sum_{j \in \mathcal{S}}\Delta_j(\mathcal{S}).
\end{multline}

This implies the coalitional value is divided equally between the NO and the set of SPs, $\phi_\text{NO} = \sum_{ j \in \mathcal{N}\setminus\{\text{NO}\}}\phi_j =\frac{v(\mathcal{N})}{2}.
$ This completes the proof.
\end{proof}

The intuition behind this equal sharing of the Shapley value between the NO and the SPs is based on~\cite{Beal2014}:{\it ``each game is decomposed into a weighted sum of unanimity games in which the Shapley value assigns an equal share of a unit to each veto player''.} In our case if the set of SPs is considered as one super-player, it is actually a veto player as well, because the value function is zero if no SP is in the coalition, since it would not be possible to collect revenues from users utilization.

\section{Application to Edge Computing}
\label{sec: numerical results}

\subsection{Parameters}
\subsubsection{Load}
\label{sec:load}

We define the load as an exogenous variable (\S\ref{sec:model}). To reproduce a realistic trend, we consider the daily traffic profile of a SP serving residential users, as modeled in~\cite{Vela2016}, i.e.
$l_t^i=a_0+\sum_{k=1}^K a_k\sin{(2k\pi\frac{t-t_k}{T})}$,
where $t$ is the timeslot and $T$ is the number of timeslots in one day; $a_k$ and $t_k$ are hyperparameters determining the amplitude and the offset of each of the $K$ sinusoidal components. 
We take their values from~\cite[Fig.2]{Vela2016}.

\subsubsection{Utility function and price}
\label{sec:utility-function-and-price}
As often observed in reality, we assume the utility~\eqref{eq:null_resources_null_utility} of any~$\text{SP}^i$ is characterized by a diminishing return effect~\cite{Salchow2008}:  the marginal utility increment becomes smaller by increasing the $h^i$. For this reason, we model the utility with the following increasing and concave function, similar to (1) of~\cite{Misra2015}:

\begin{equation}
\label{eq: utility_func}
    u^i(l_t^i, h^i)=\beta^i \cdot l^i_t \cdot (1- \emph{e}^{- \xi \cdot h^i}).
\end{equation}

The term $\beta^i$ is the \emph{benefit factor} of player $i$ which represents the benefit that a SP gets from serving one unit of load at the Edge. It is a multiplicative constant, null for the NO, $\beta^\text{NO}= 0$. The term $\xi$ models the shape of the diminishing return, i.e. how fast it saturates to its upper bound $\beta^i \cdot l^i_t$. Note that this utility function follows property~\eqref{eq:null_resources_null_utility}.

\subsection{Scenario with 2 SPs of the same type}
\label{sec:secnatio2SPssame}

In this case there are two SPs of the same type: $\beta^{\text{SP}^1}=\beta^{\text{SP}^2}=\hat p$ where $\hat p \triangleq\frac{d}{D \cdot T}$ is the price, $d=0.05$~\emph{dollars/millicores}, amortized over each of $T$ time slots over the investment duration, $\text{D}$. $\text{SP}^1$ and $\text{SP}^2$ have the same temporal trends, but $l_t^1=4l_t^2 \ \forall t$.

In Fig.~\ref{fig: twoSPs capacity and coalitional value}, we show the capacity of purchased CPU and the value of the grand coalition, as a function of the daily total load, $l^\text{tot}=\sum_{t=1}^T l_t^\text{tot}$. We observe that, the more the load the more the capacity installed to serve it. However, recall that the utility functions follow a diminishing return with respect to the resources, so, the trend of the capacity $C$ is sublinear. We observe a linear trend for coalitional value because the value function is linearly dependent on the load (see Eqn.~\eqref{eq: utility_func}).

\begin{figure}[t]
     \centering
 {\includegraphics[width=0.44\textwidth, height=90pt]{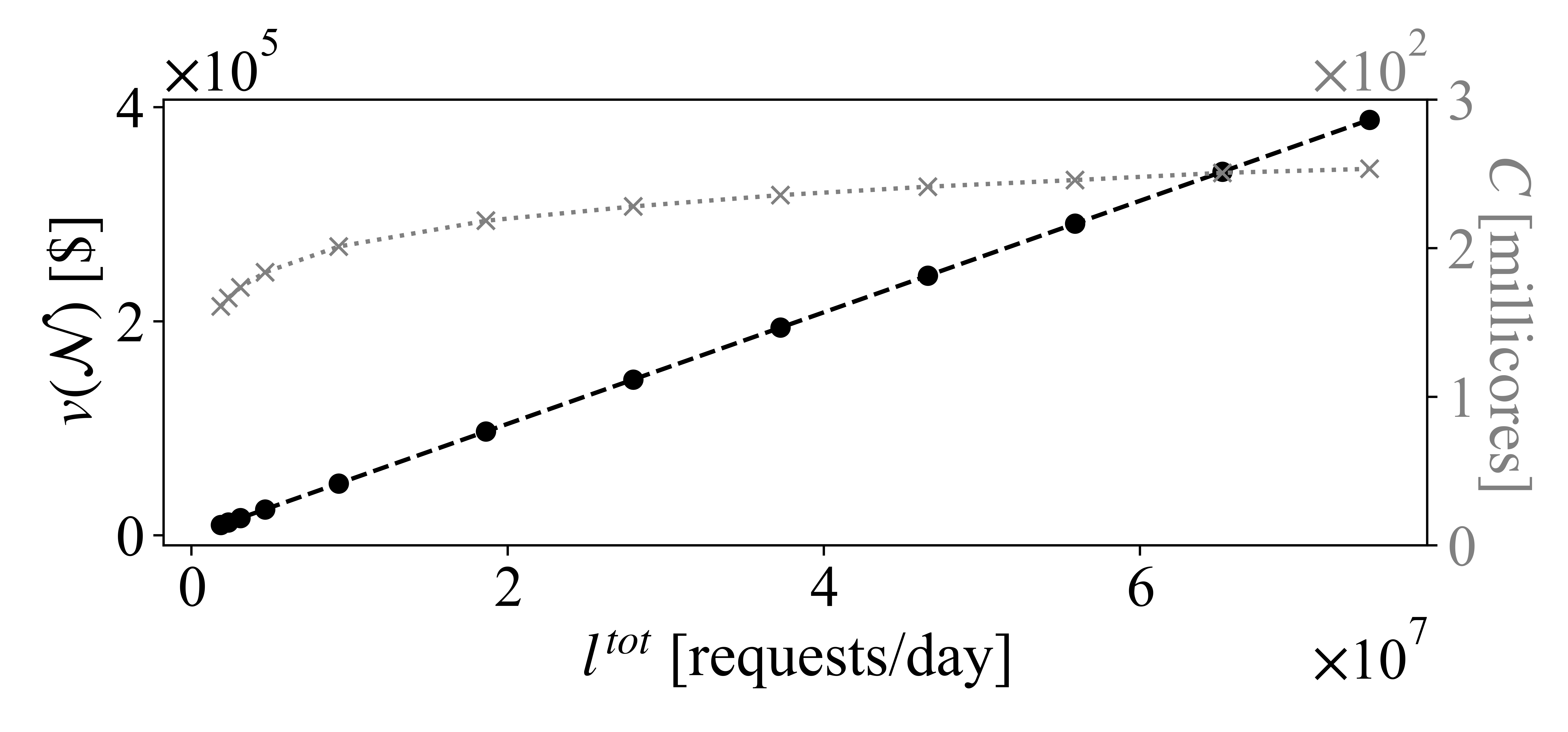}}
        \caption{Capacity and coalitional value as a function of the overall daily load}
        \label{fig: twoSPs capacity and coalitional value}
\end{figure}

 We observe in Fig.~\ref{fig: twoSPs capacity split and revenues contribution} capacity sharing between the SPs, $\text{SP}^1$ receives a larger capacity: it has to serve a larger part of the requests. Note that, even if the load of $\text{SP}^1$ is 4 times the load of $\text{SP}^2$, the difference between the resource allocated to them is not that big: a consequence of the diminishing return.

 The contribution of $\text{SP}^i$ to the coalitional revenues is defined as $\hat r^i=D \cdot \sum_{t=1}^T u^i(l_t^i,h^i)$.
 We denote the grand coalitional revenues as $r^\mathcal{N}=\sum_{i\in\mathcal{N}} r^i$, where $r^i$ is the result obtained in~\eqref{eq:payments1}-\eqref{eq:payments3}, i.e. the payoff of each player without considering the component of the payment. The term $\hat r^i$ is the amount of revenues produced by SP $i$, due to the served load during the overall duration of the coinvestment.
 
Fig.~\ref{fig: twoSPs capacity split and revenues contribution} shows that most contribution comes from $\text{SP}^1$, since its load is four times higher than that of $\text{SP}^2$, and the utility of any SP (and thus its contribution to the grand coalitional revenues) is proportional to the served load; indeed, we observe that $\hat r^{\text{SP}^1}= \frac{1}{4}\hat r^{\text{SP}^2}$.
Note that $h^\text{NO}$ and $r^\text{NO}$ are not in the figure, as the NO does not use resources, because its load is null; this implies that its utility is null so it does not produce revenues to the grand coalition by serving a load.
\begin{figure}[t]
     \centering
 {\includegraphics[width=0.44\textwidth]{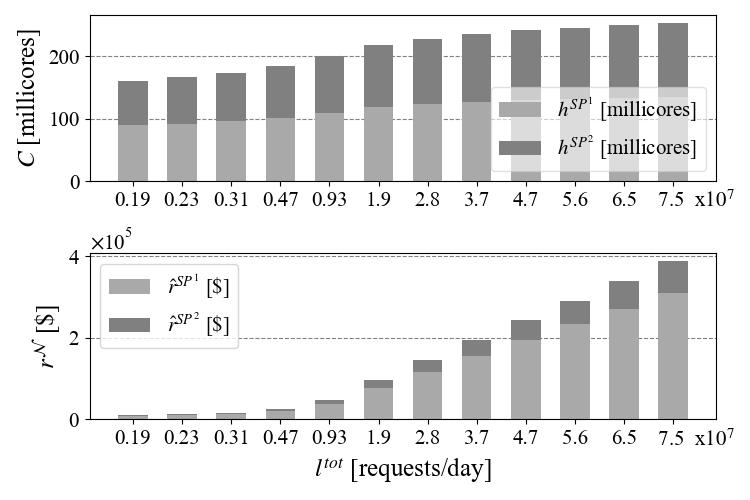}}
        \caption{Split of capacity and contributions to coalitional value}
        \label{fig: twoSPs capacity split and revenues contribution}
\end{figure}

\begin{figure}
     \centering
     \begin{subfigure}[b]{0.44\textwidth}
         \centering
        {\includegraphics[width=\textwidth]{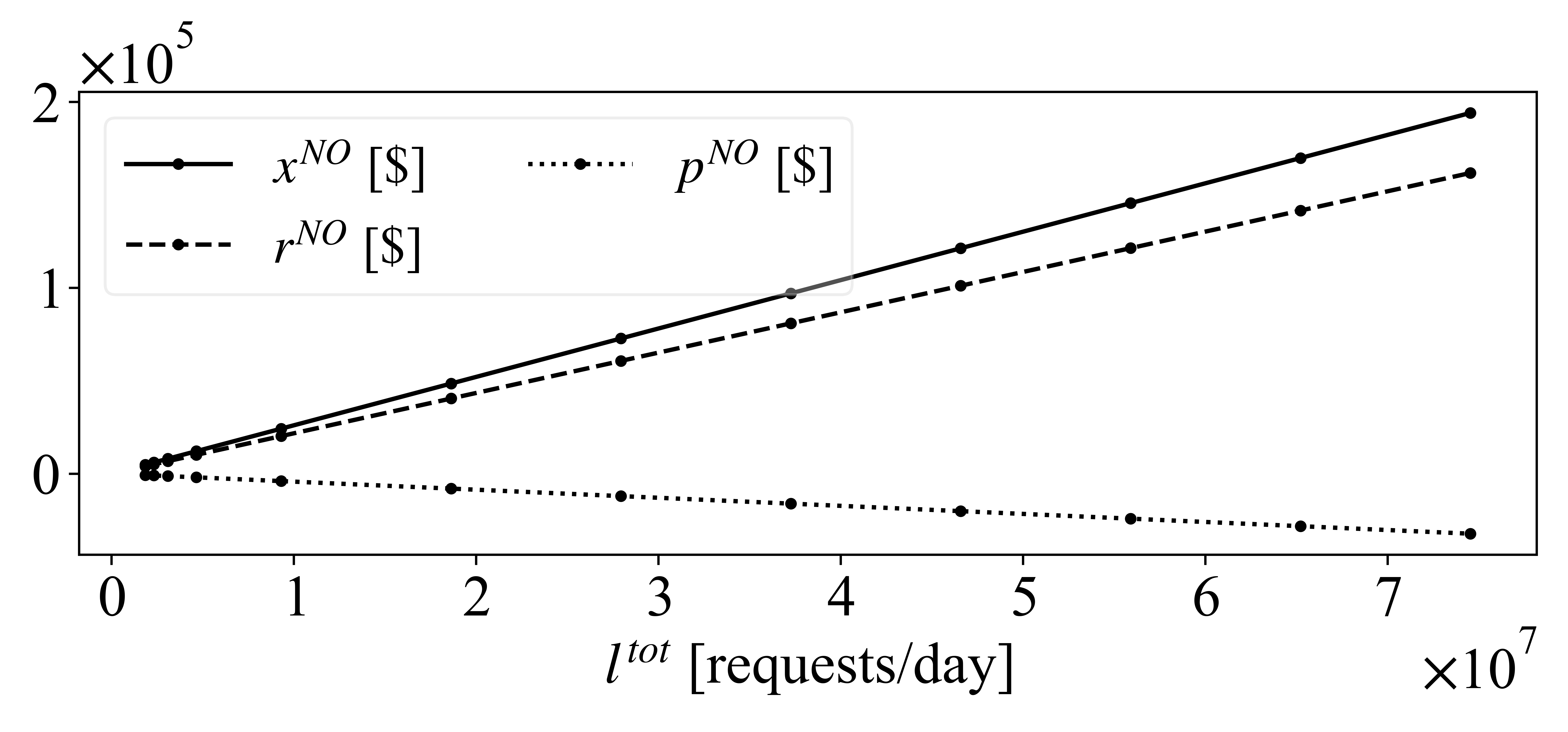}}
         \caption{Payoff, revenues and payment by NO}
         \label{fig: two Sps shapley NO}
     \end{subfigure}
     
          \begin{subfigure}[b]{0.44\textwidth}
         \centering
        {\includegraphics[width=\textwidth]{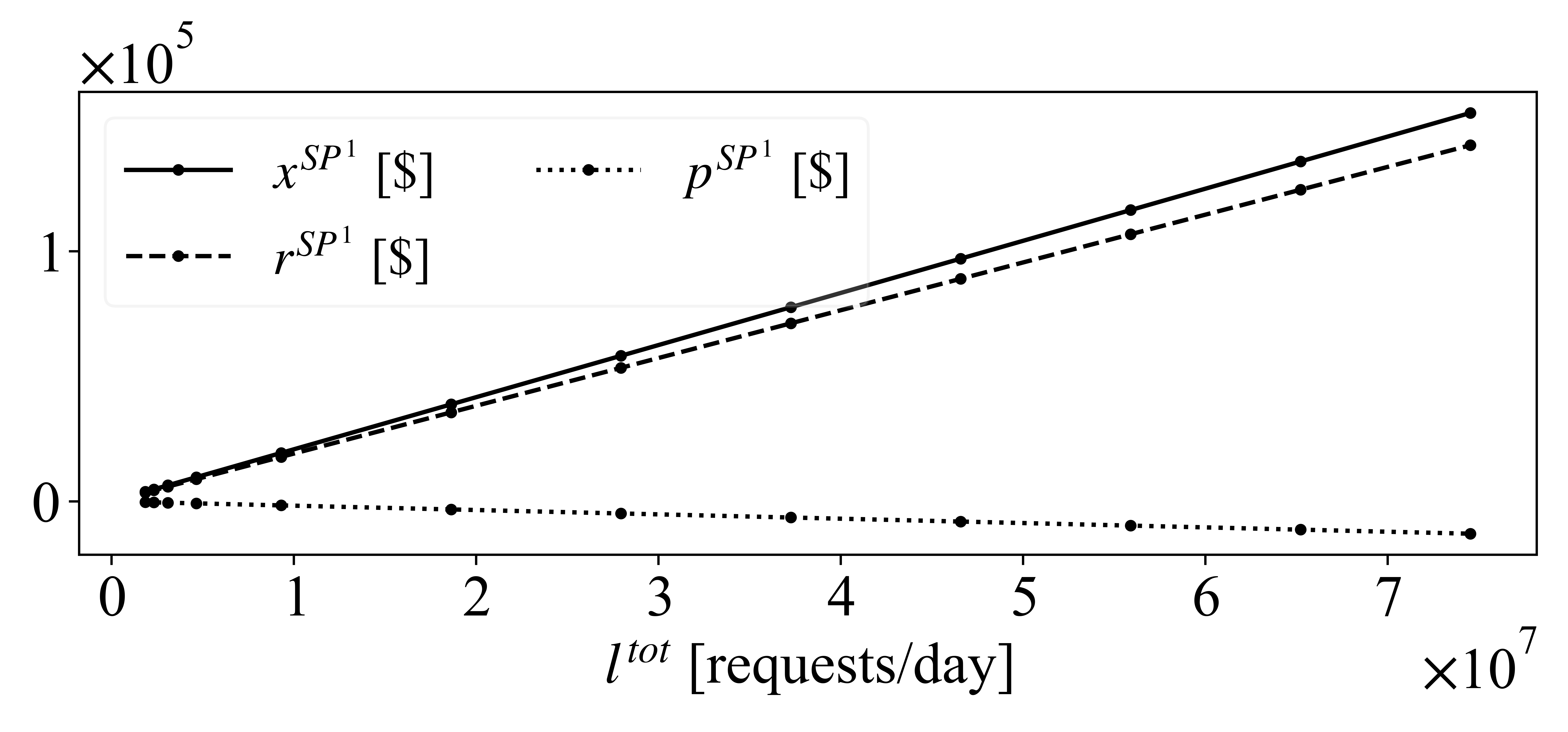}}
         \caption{Payoff, revenues and payment by $\text{SP}^1$}
         \label{fig: two Sps shapley SP1}
     \end{subfigure}
     
          \begin{subfigure}[b]{0.44\textwidth}
         \centering
        {\includegraphics[width=\textwidth]{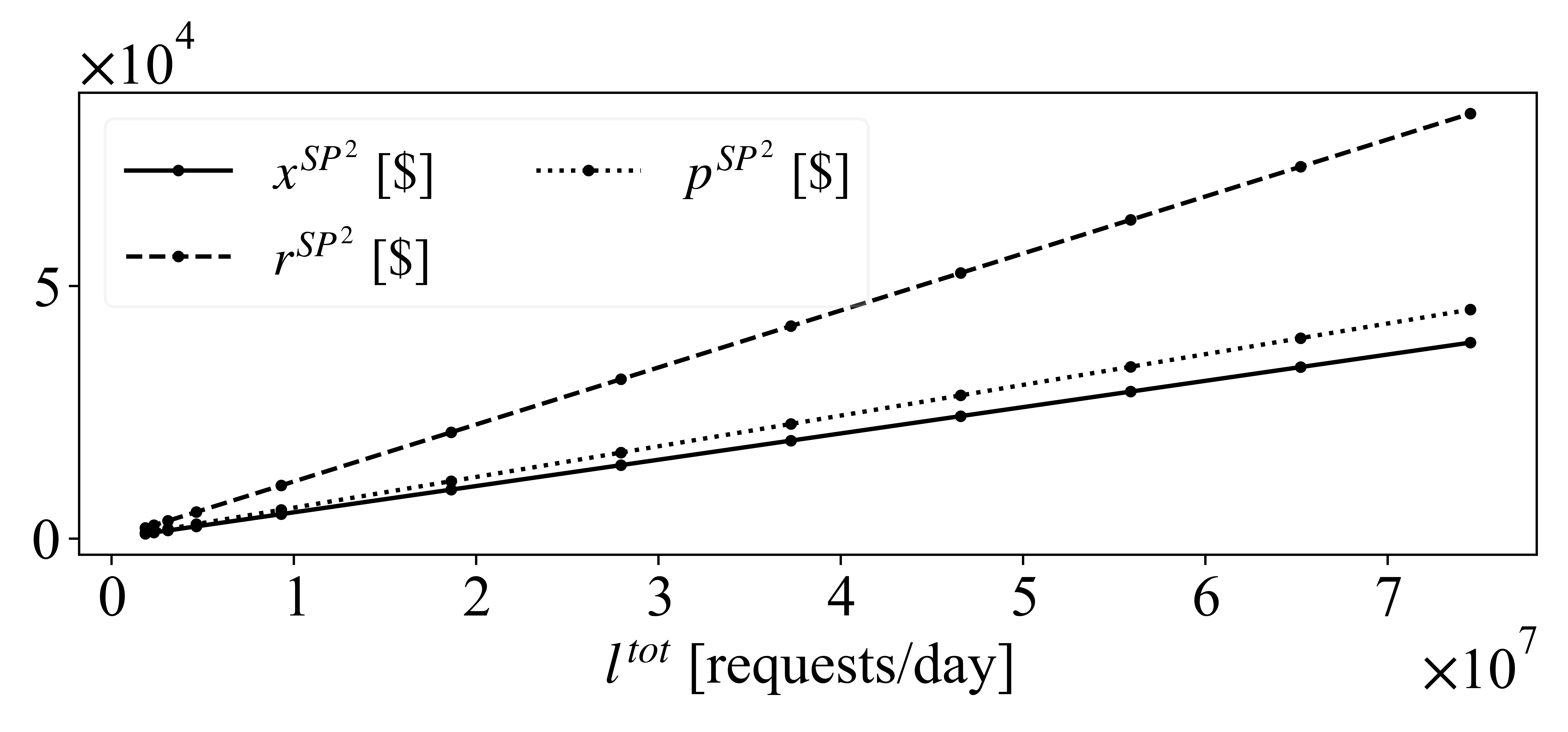}}
     \caption{Payoff, revenues and payment by $\text{SP}^2$}
     \label{fig: two Sps shapley SP2}
    \end{subfigure}
    \caption{Shapley value: payoffs, revenues and payments.}
    \label{fig: twoSPs payoffs Shapley}
\end{figure}

Fig.~\ref{fig: twoSPs payoffs Shapley} shows the outcome of the game, i.e. the payment $p^i$, the revenue $r^i$ and the payoff $x^i$ of each player $i\in\mathcal{N}$ given by the Shapley value. We first observe that the payoff of anyone increases with the total load, which is obvious as the revenues of the grand coalition is the sum of the utilities of each player, which in turn increase with the number of served requests. It is interesting to notice that only $\text{SP}^2$ actually pays to deploy the resources at the Edge, while the $\text{NO}$ and $\text{SP}^1$ have negative payments, so, they are not paying, but are being paid. In the case of $\text{SP}^1$, this means that it enjoys an additional gain from the coinvestment in the Edge resources, which sums to the revenue $r^{\text{SP}^1}$ directly coming from its customers.  The ``privilege'' of the NO and $\text{SP}^1$ can be explained by the fact that they are the most important for the coalition: NO is the veto player; $\text{SP}^1$ brings to the coalition most of the revenues collected from users (Fig.~\ref{fig: twoSPs capacity split and revenues contribution}).

\subsection{Scenario with 2 SPs of different types}
In this case we have two SPs offering different types of services. For instance, $\text{SP}^2$ may offer an extremely low-latency service, e.g. augmented reality, while $\text{SP}^1$ may offer a less stringent service, e.g. online gaming. In this case, the utility coming from serving a request at the Edge is much higher for $\text{SP}^2$ than for $\text{SP}^1$, since the latter could serve some of the requests (for instance requests not related with interactions with the player) from the Cloud without degrading too much the perceived user experience. Therefore, we consider now that $\beta^{\text{SP}^2}\ge \beta^{\text{SP}^1}$. As in the previous scenario, $\beta^\text{tot}=\beta^{\text{SP}^1}+\beta^{\text{SP}^2}=2\hat p$. 
We assume further that $\beta^{\text{SP}^1}=(1-\omega)\cdot \beta^\text{tot}$ and $\beta^{\text{SP}^2}=\omega\cdot \beta^\text{tot}$. We make $\omega$ vary in $[0.5,1]$. The case $\omega=0.5$ corresponds to the previous scenario, where the SPs were of the same type. Increasing $\omega$, the two SPs become more heterogeneous, and in particular $\text{SP}^2$ has higher benefits per unit of load than $\text{SP}^1$. We keep the load as before. 

We now show how resource allocation and payoff change with $\omega$. We observe in Fig.~\ref{fig: two Sps shapley compact changing omega avg capacity} that increasing $\omega$, the percentage of CPU given to the $\text{SP}^1$ decreases and is null for $\omega=1$, at which point it is given entirely for $\text{SP}^2$. This is due to the fact that, despite it attracts most of the user load, $\text{SP}^1$ is less useful to assign resources to, as its benefit factor becomes smaller with $\omega$. This tells us, as expected, that resource allocation at the Edge must be taken not only based on load, but also on the nature of the services, and in particular on time-sensitivity, which is reflected in a different benefit per unit of load satisfied at the Edge. In Fig.~\ref{fig: two Sps shapley omega} the payoff sharing reflects what we mentioned above. The marginal contribution brought by $\text{SP}^2$ increases since it produces most of the coalition revenues. In all the cases, the NO has 1/2 the coalitional value, (Th.~\ref{theo: payoff_proportions}).  

\begin{figure}
     \centering
     \begin{subfigure}[b]{0.4\textwidth}
         \centering
        {\includegraphics[width=0.9\textwidth]{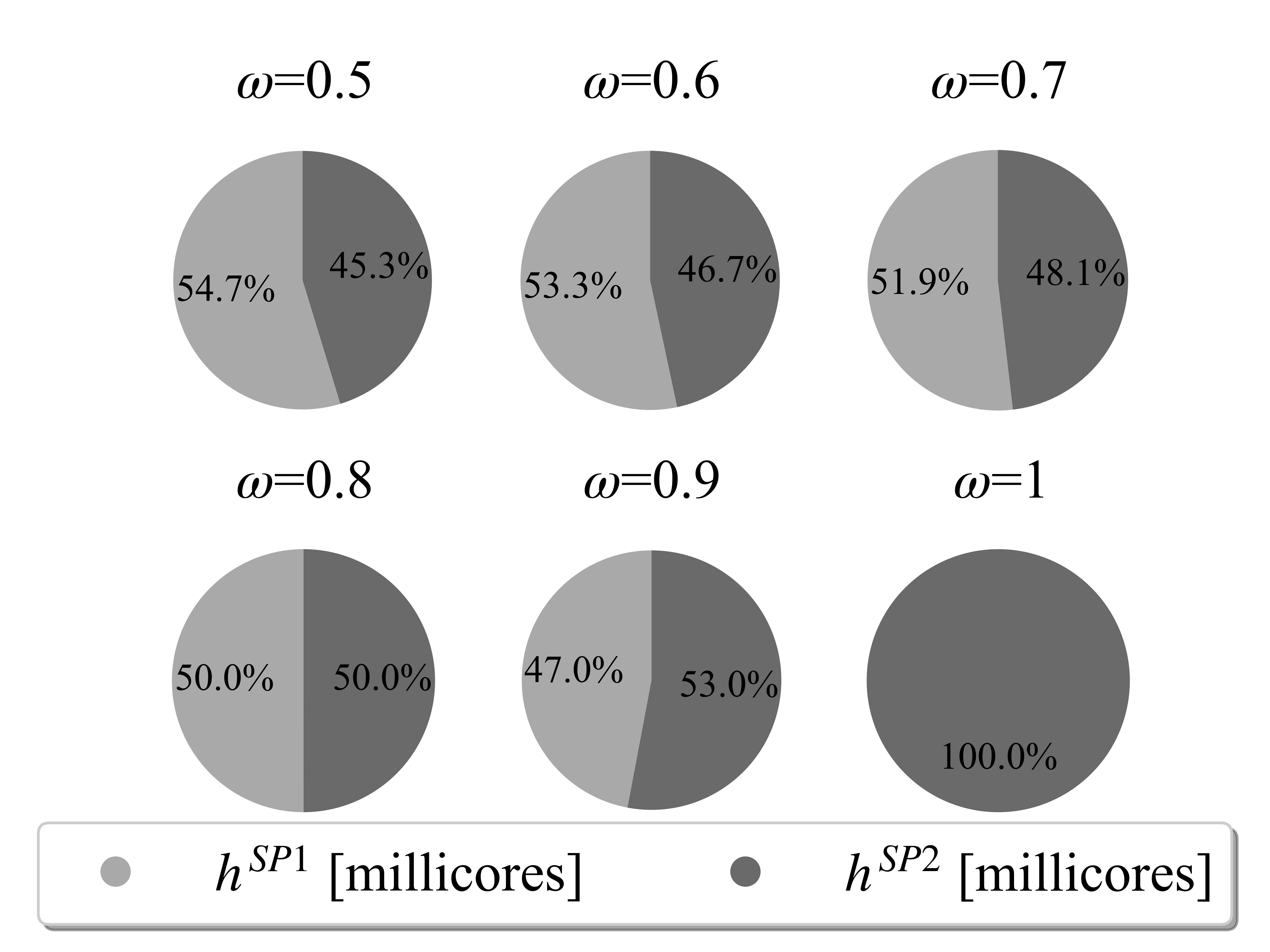}}
         \caption{Capacity subdivision among the players}
         \label{fig: two Sps shapley compact changing omega avg capacity}
     \end{subfigure}
     
          \begin{subfigure}[b]{0.4\textwidth}
         \centering
        {\includegraphics[width=0.9\textwidth]{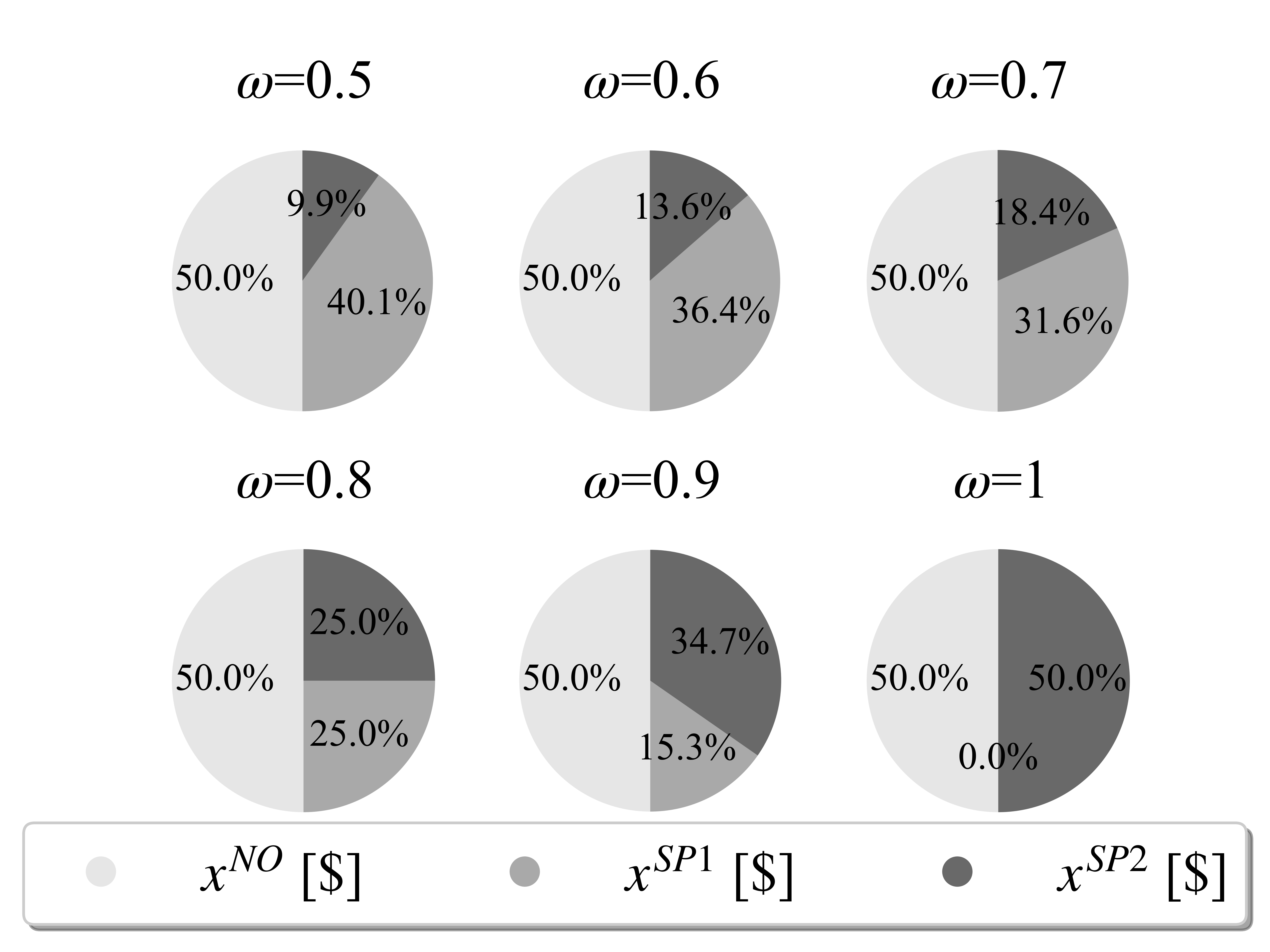}}
     \caption{Payoff sharing}
     \label{fig: two Sps shapley omega}
    \end{subfigure}
    \caption{Coalitional value and capacity function of $\omega$}
    \label{fig: two Sps shapley compact changing omega}
\end{figure}

\subsection{Price sensitivity analysis}

Here, we assess our model to show its behavior varying the number of SPs.
The consequence of increasing the price is that (Fig.~\ref{fig:manySpsSensitivity}) (i) the purchased capacity is reduced but in a sub-linear way, because of~\eqref{eq: utility_func} and the coalitional value decreases linearly. These trends remain consistent when changing the number $N$ of SPs. Fig.~\ref{fig:manySpsSensitivity} also confirms that adding a player in the game, brings a higher benefit in terms of the value of the coalition $v$. This is in line with the  supermodularity of $v$ and hence the convexity of the game (Th.~\ref{theo: the game is convex}) and its stability.

\begin{figure}
     \centering
     \begin{subfigure}[b]{0.4\textwidth}
         \centering
        {\includegraphics[width=0.85\textwidth]{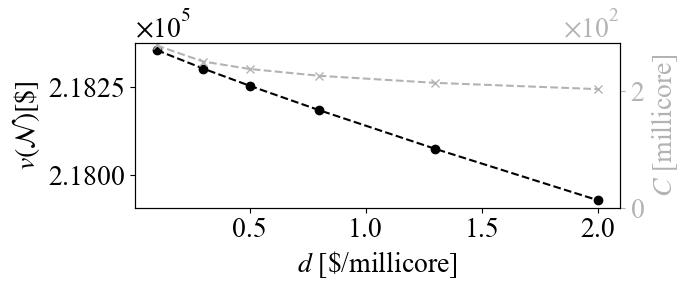}}
         \caption{NO and $N=2$ SPs}
         \label{fig:2SpsSensitivity}
     \end{subfigure}
     
          \begin{subfigure}[b]{0.4\textwidth}
         \centering
        {\includegraphics[width=0.85\textwidth]{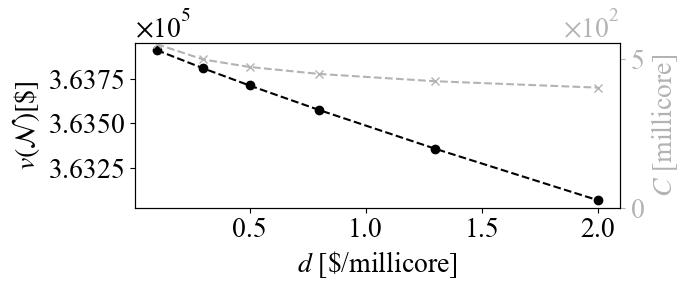}}
     \caption{NO and $N=4$ SPs}
     \label{fig:4SpsSensitivity}
    \end{subfigure}
      \begin{subfigure}[b]{0.4\textwidth}
         \centering
        {\includegraphics[width=0.85\textwidth]{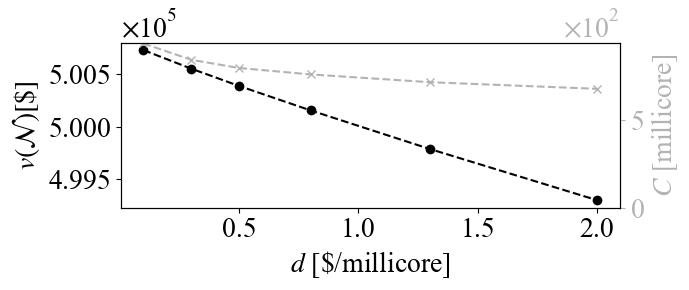}}
     \caption{NO and $N=7$ SPs}
     \label{fig:7SpsSensitivity}
    \end{subfigure}
    \caption{Coalitional value and capacity function of the price $d$ for different numbers of players}
    \label{fig:manySpsSensitivity}
\end{figure}

\section{Conclusion}
\label{sec: conclusion}
We proposed a coalitional game theory solution to enable coinvestment between heterogeneous players (NO and SPs) 
and applied it to the of deployment of Edge Computing. By showing the convexity of the game, we proved that the core is non-empty and that the Shapley value, which provides a fair way to divide income among players, lies in the core. So, it is always possible to form the grand coalition, made of all players. We
%, e.g. that the NO earns half of the coalitional value and the other half goes to the set of SPs. In the numerical results, we 
studied numerically the solution under different scenarios. For future work, we will consider adding a strategy-proof enforcement feature to ensure that players are truthful.


\vspace{-0.9cm}
\begin{thebibliography}{00}
\vspace{-0.2cm}
%\bibitem{aws} 
\bibitem{Guillemin2021} F. Guillemin and V. Rodriguez,  ``Evaluating the impact of tower companies on the telecommunications market''., IEEE ICIN, 2021.
\bibitem{AraldoSAC2020} A. Araldo, A. Di Stefano et al. ``Resource Allocation for Edge Computing with Multiple Tenant Configurations'', ACM SAC, 2020
\bibitem{Araldo2016} A. Araldo, G. Dán, D. Rossi, ``Stochastic Dynamic Cache Partitioning for Encrypted Content Delivery'', ITC 2016
\bibitem{Sabella2019}
D. Sabella et al., "Edge Computing: from standard to actual infrastructure deployment and software development", ETSI White paper, 2019.
\bibitem{Afraz2018}
N. Afraz and M. Ruffini, ``A Sharing Platform for Multi-Tenant PONs''. Journal of Lightwave Technology, 2018.
\bibitem{Kiedanski2020}
D. Kiedanski, A. Orda, and D. Kofman, "Discrete and stochastic coalitional storage games", ACM e-Energy, 2020.
\bibitem{Tamer2009}
W. Saad, Z. Han, M. Debbah, A. Hjorungnes and T. Basar, "Coalitional game theory for communication networks", IEEE Sig. Proc. Mag., 2009.
\bibitem{Salchow2008} Jr. K. J. Salchow  "Clustered Multiprocessing: Changing the Rules of the Performance Game", F5 White Paper, 1-11, 2008.
\bibitem{Sereno2012}M. Sereno,``Cooperative game theory framework for energy efficient policies in wireless networks'',  Intern. Conf. on Future Systems, 2012
\bibitem{Moura2019} J. Moura et al., ``Game theory for multi-access edge computing: Survey, use cases, and future trends'', in IEEE Commun. Surveys Tuts, 2019
\bibitem{Pham2019} Q.V.Pham et al. ``Coalitional games for computation offloading in NOMA-enabled multi-access edge computing'', IEEE Tr.Veh.Tech. 2019
\bibitem{Osborne1994}  M. J. Osborne and A. Rubinstein, "A Course in Game Theory", Massachusetts Institute of Technology, 257-275, 1994.
\bibitem{Shapley1971}
L. S. Shapley, "Cores of convex games", Int.J.Game Th. 1971
\bibitem{Misra2015} R. T. B. Ma, J. C. S. Lui and V. Misra, ``Evolution of the Internet Economic Ecosystem,'' in IEEE/ACM Tr.on Net., 2015
\bibitem{Driessen1988}  T. Driessen, "Cooperative Games, Solution and Applications", Kluwer Academic Publisher, 1988.
\bibitem{Beal2014}
B. Sylvain, E. Rémila and P. Solal, "Veto players, the kernel of the Shapley value and its characterization", pre-print, 2014.
\bibitem{Vela2016} A. P. Vela, A. Vía, F. Morales, M. Ruiz and L. Velasco, "Traffic generation for telecom cloud-based simulation", ICTON, 2016.
\bibitem{Brink2007} R. van den Brink, ``Null or nullifying players: the difference between the Shapley value and equal division solutions '', J. of Econ. Theo., 2007
\bibitem{Hao2021} J. Koch and W. Hao, ``An Empirical Study in Edge Computing Using AWS'', IEEE CCWC, 2021
\end{thebibliography}
\end{document}